\documentclass[10pt]{article}
\usepackage{enumitem}
\usepackage{amsfonts}
\usepackage{amsmath}
\usepackage{graphicx}
\usepackage{subfigure}
\usepackage{fancyhdr}
\usepackage{amsthm}
\usepackage{amssymb}
\usepackage{amscd}
\usepackage{float}
\newtheorem{lemma}{Lemma}[section]
\newtheorem{proposition}{Proposition}[section]
\newtheorem{remark}{Remark}[section]

\begin{document}
\title{{Volatility swaps valuation under stochastic volatility with jumps and stochastic intensity}\thanks{This work was supported by  the National Natural Science Foundation of China (11471230, 11671282).}}
\author{{Ben-zhang Yang$^a$, Jia Yue$^b$, Ming-hui Wang$^a$ and Nan-jing Huang$^a$\thanks{Corresponding author.  E-mail address: nanjinghuang@hotmail.com}}\\
{\small\it a. Department of Mathematics, Sichuan University, Chengdu, Sichuan 610064, P.R. China}\\
{\small\it b. Department of Economic Mathematics, South Western University of Finance and Economics,}\\
{\small\it Chengdu, Sichuan 610074, P.R. China}}
\date{}
\maketitle
\vspace*{-9mm}
\begin{center}
\begin{minipage}{5.8in}
{\bf Abstract.}
 In this paper, a pricing formula for volatility swaps is delivered when the underlying asset follows the stochastic volatility model with jumps and stochastic intensity. By using Feynman-Kac theorem,  a partial integral differential equation is obtained to derive the joint moment generating function of the previous model.
 Moreover,  discrete and continuous sampled volatility swap pricing formulas are given by employing transform techniques and the relationship between two pricing formulas is discussed. Finally, some numerical simulations are reported to support the results presented in this paper.

\end{minipage}
\end{center}
\section{Introduction}

Since the finance volatility has always been considered as a key measure, the development and growth of the financial market have changed the role of the volatility over last century. Volatility derivatives in general are important tools to display the market fluctuation and manage volatility risk for investors. More precisely, volatility derivatives are traded for decision-making between long or short positions, trading spreads between realized and implied volatility, and hedging against volatility risks.  The utmost advantage of volatility derivatives is their capability in providing direct exposure towards the assets volatility without being burdened with the hassles of continuous delta-hedging. Various theoretical results, numerical algorithms and applications have been studied extensively for volatility derivatives in the literature (see, for example, \cite{Carr98,Carr,Carr2,Pun,Santa,Windcliff,Zheng,Zhu11,Zhu15}).

With the rapid growth of trading of variance/volatility swaps in the past twenty years, researchers in this field attempt to construct more practical models and find more feasible methods for pricing variance/volatility swaps. Incorporating jump diffusions into models of pricing and hedging variance swaps, Carr et al. \cite{Carr2} and Huang et al. \cite{Huang} studied the existence of many small jumps that cannot be adequately modelled by using finite-activity compound Poisson processes. Cont and Kokholm \cite{Cont13} presented a model for the joint dynamics of a set of forward variance swap rates
along with the underlying index; they used L\'evy processes as building blocks and provided the tractable pricing framework for variance swaps, VIX futures, and vanilla call/put options. Zhu and Lian \cite{Zhu11} solved the discretely sampled variance swaps pricing formula under Heston's stochastic volatility model using a partial differential equation approach. Recently, Yang et al. \cite{Yang18} focus on the pricing of the variance swaps in the financial market where the stochastic interest rate and the volatility of the stock are driven by the Cox-Ingersoll-Ross model and Heston model with simultaneous L\'evy jumps, respectively.  However, to our best knowledge, there are only a few researchers to consider the pricing of volatility swaps in the literature. Recently, by employing the method for pricing discretely-sampled variance swaps in \cite{Zhu11},  Zhu and Lian \cite{Zhu15} studied analytical valuation for volatility swaps under the framework of the Heston stochastic volatility model;  they obtained a closed-form exact solution for discretely-sampled volatility swaps with the realized volatility defined as the average of the absolute percentage increment of the underlying asset price.

It is well known that jumps \cite{Cont04} and the stochastic intensity are important features of financial assets in pricing derivatives. Bates \cite{Bates} proposed the generalized model with jumps to highlight the impact of the jumps in the underlying asset. Lian and Zhu \cite{Zhu11} extended the underlying asset price process allowing the stochastic volatility with simultaneous jumps (SVSJ) model and pointed out that such a model can describe the real market better than previous one. Santa-Clata and Yan \cite{Santa} found the components of the jump intensity risk after calibrating the S\&P 500 index option prices from the beginning of 1996 to the end of 2002. Huang et al. \cite{Huang} studied the valuation of the option when the underlying asset follows the double exponential jump process with the stochastic volatility and stochastic intensity, in which the model captures the stock prices, volatility and stochastic intensity. However, as pointed out by Chang et al. \cite{Chang},  it is necessary to make the additional extension work to incorporate both the varying continuous volatility and stochastic intensity to model the characteristic of the asset price in the financial market.

The main purpose of this paper is to make an attempt to propose a new stochastic volatility model with jumps and stochastic intensity and consider the volatility swaps valuation problem described by this model. The rest of this paper is organized into five sections. We shall start with a description of our model in Section 2. Then, volatility pricing formulas for discrete and continuous samples are given in Section 3. Some numerical examples are reported in Section 4. The paper ends with conclusions in Section 5.

\section{Stochastic volatility model with jumps and stochastic intensity framework}

Let $(\Omega,\mathcal{F},\mathbb{P})$ be a probability space with a risk neutral probability $\mathbb{P}$.  Assume that the underlying asset price $S=(S_t)_{t\in [0,T]}$ with a instantaneous squared volatility $V=(V_t)_{t\in [0,T]}$ and the jump intensity process $\lambda=(\lambda_t)_{t\in [0,T]}$ of Possion process $N_t$ can be governed by the following system of SDEs:
\begin{eqnarray}\label{real}
\left\{
\begin{array}{ll}
dS_t=(r-d-\lambda_tm)S_tdt+\sqrt{V_t}S_tdW_t^S+(e^{J^S}-1)S_tdN_t,\\
dV_t=\kappa_V(\theta_V-V_t)dt+\sigma_V\sqrt{V_t}dW_t^V+J^VdN_t,\\
d\lambda_t=\kappa_{\lambda}(\theta_{\lambda}-\lambda_t)dt+\sigma_{\lambda}\sqrt{\lambda_t}dW_t^{\lambda},
\end{array}
\right.
\end{eqnarray}
where
$W^S=(W_t^S)_{t\in [0,T]}$ and $W^V=(W_t^V)_{t\in [0,T]}$ are correlated Brownian motions with a constant correlation coefficient such that the quadratic covariation between $W^S$ and $W^V$ satisfies $d [W^S,W^V]_t=\rho dt$ for some constant $\rho\in [-1,1]$, and $N_t$ is independent with $W^S=(W_t^S)_{t\in [0,T]}$ and $W^V=(W_t^V)_{t\in [0,T]}$; $r$ and $d$ denote the riskless interest rate and the constant dividend yield, respectively; $J^S$ and $J^V$ denote the jump sizes of the price and variance, respectively, in which the jump sizes are assumed to be independent with $W^S$, $W^V$ and $N_t$.  In addition, we assume that $m$ is the average jump amplitude of price with $m=E^{\mathbb{Q}}[e^{J^S}-1]$,  the mean-reverting speed parameters $\kappa_V$ and $\kappa_\lambda$ are positive constants,  the positive constants $\sigma_V$ and $\sigma_V$ are long term volatilities of $V_t$ and $\gamma_t$, respectively, the long term means $\theta_V$ and $\theta_\gamma$ are constants such that $2\kappa_V\theta_V>\sigma_V$ and $2\kappa_\lambda\theta_\lambda>\sigma_\lambda$, and the Brownian motion $W^\lambda=(W_t^\lambda)_{t\in [0,T]}$ is independent of $W^S$ and $W^V$.

We would like to point out that \eqref{real} includes several known models as special cases. In fact, if stochastic intensity process $\gamma(t)$ is a constant, then system \eqref{real} reduces to the model considered by Zhu and Lian \cite{Zhu11}, Zheng and Kwok \cite{Zheng}. Moreover,  if there the jump diffusion is removed, then system \eqref{real} reduces to the model considered by Huang et al. \cite{Huang}, and if stochastic intensity process $\gamma(t)$ is a constant and there is no jump diffusion, then system \eqref{real} reduces to the model considered by Carr et al. \cite{Carr98} and Zhu et al. \cite{Zhu15}.

Proposition \ref{theo} below completes the characteristic of the joint moment-generating function (short for MGF) of the joint processes $X_t$, $V_t$ and $\lambda_t$, in which system (\ref{real}) is converted to a forward log-asset price system by using It\^{o} lemma and MGF can be obtained by solving a partial differential equation
after employing the Feynman-Kac theorem.

\begin{proposition}\label{theo}
Let $X_t=\ln S_t$ be the log-price process. Then MGF of the joint processes $X_t$, $V_t$ and $\lambda_t$ can be defined as follows:
$$U(t,X,V,\lambda):=\mathbb{E}^{\mathbb{Q}}\left[\exp(\omega X_T+ \varphi V_T+ \psi \lambda_T+\chi)|X(t)=X,V(t)=V,\lambda(t)=\lambda\right],$$
where $\varphi$, $\psi$ and $\chi$ are constant parameters. Moreover, if
$$
\mathbb{E}^{\mathbb{Q}}\left[\exp(\omega X_T+ \varphi V_T+ \psi \lambda_T+\chi)\right]<\infty,
$$
then the value of $U(\tau,X,V,\lambda)$ at $\tau := T-t$ can be given as follows
$$U(\tau,X,V,\lambda)=\exp(\omega X+C(\tau;q)V+D(\tau;q)\lambda+E(\tau;q)),$$
where $q=(\omega,\varphi,\psi,\chi)$ and
$C(\tau;q),D(\tau;q),E(\tau;q)$ satisfy
\begin{eqnarray}\label{odes}
\left\{
\begin{aligned}
&\frac{d C(\tau;q)}{d \tau}=\frac{1}{2}\sigma_V^{2}C^{2}(\tau;q)+(\rho\sigma_V\omega-\kappa_V)C(\tau;q)+\frac{1}{2}(\omega^2-\omega)\\
&\frac{d D(\tau;q)}{d\tau}=\frac{1}{2}\sigma^2_\lambda D^{2}(\tau;q)-\kappa_\lambda D(\omega,\tau)+\Lambda(\tau;q)\\
&\frac{d E(\tau;q)}{d\tau}=(r-d)\omega+\kappa_V\theta_V C(\tau;q)+\kappa_\lambda\theta_\lambda D(\tau;q)
\end{aligned}
\right.
\end{eqnarray}
with initial conditions
\begin{equation}\label{conditions}
C(0;q)=\varphi,\quad D(0;q)=\psi,\quad E(0;q)=\chi,
\end{equation}
where
\begin{eqnarray}\label{e+}
\Lambda(\tau;q)=-m\omega+E^{\mathbb{Q}}\left[\exp(\lambda J^S+CJ^V)-1\right].
\end{eqnarray}

\end{proposition}
\begin{proof}
From (\ref{real}), we have the following stochastic differential equation with respect to $X(t)$, $V(t)$ and $\lambda(t)$:
\begin{eqnarray}\label{neural}
\left\{
\begin{array}{ll}
dX_t=(r-d-\lambda_tm-\frac{1}{2}V_t)dt+\sqrt{V_t}dW_t^S+(e^{J^S}-1)dN_t,\\
dV_t=\kappa_V(\theta_V-V_t)dt+\sigma_V\sqrt{V_t}dW_t^V+J^VdN_t,\\
d\lambda_t=\kappa_{\lambda}(\theta_{\lambda}-\lambda_t)dt+\sigma_{\lambda}\sqrt{\lambda_t}dW_t^{\lambda}.
\end{array}
\right.
\end{eqnarray}

Thus, by applying It\^{o} lemma to $U(t,X_t,V_t,r_t)$, we can obtain a partial integral-differential equation (PIDE) for $U(t,X,V,\lambda)$ as follows
\begin{eqnarray}\label{PIDEt}
0&= & \frac{\partial U}{\partial t}+(r-d-\lambda m-\frac{1}{2}V)\frac{\partial U}{\partial X}+[\kappa_V(\theta_V-V)]\frac{\partial U}{\partial V}+[\kappa_\lambda(\theta_\lambda-\lambda)]\frac{\partial U}{\partial \lambda}\nonumber\\
&& \mbox{} +\frac{1}{2}V\frac{\partial^{2} U}{\partial X^{2}}
+\rho \sigma_V V\frac{\partial^{2} U}{\partial X \partial V}
+\frac{1}{2}\sigma_V^{2}V\frac{\partial^{2} U}{\partial V^{2}}
+\frac{1}{2}\sigma_{\lambda}^{2}\lambda\frac{\partial^{2} U}{\partial \lambda^{2}}\nonumber\\
&& \mbox{} +\lambda E^{\mathbb{Q}}\left[U(X+J^S,V+J^V,\lambda,t)-U(X,V,\lambda,t)\right].
\end{eqnarray}
Denoting $\tau=T-t$, we get
\begin{eqnarray}\label{PIDEt}
\frac{\partial U}{\partial \tau}&= & (r-d-\lambda m-\frac{1}{2}V)\frac{\partial U}{\partial X}+[\kappa_V(\theta_V-V)]\frac{\partial U}{\partial V}+[\kappa_\lambda(\theta_\lambda-\lambda)]\frac{\partial U}{\partial \lambda}\nonumber\\
&& \mbox{} +\frac{1}{2}V\frac{\partial^{2} U}{\partial X^{2}}
+\rho \sigma_V V\frac{\partial^{2} U}{\partial X \partial V}
+\frac{1}{2}\sigma_V^{2}V\frac{\partial^{2} U}{\partial V^{2}}
+\frac{1}{2}\sigma_{\lambda}^{2}\lambda\frac{\partial^{2} U}{\partial \lambda^{2}}\nonumber\\
&& \mbox{} +\lambda E^{\mathbb{Q}}\left[U(X+J^S,V+J^V,\lambda,\tau)-U(X,V,\lambda,\tau)\right].
\end{eqnarray}

Due to the affine structure in the SVSJ model, (\ref{PIDEt}) admits an analytic solution of the following form:
\begin{equation}\label{affine}
U(\tau,X,V,\lambda)=\exp(\omega X+C(\tau;q)V+D(\tau;q)\lambda+E(\tau;q))
\end{equation}
with the initial condition
$$U(0,X,V,\lambda)=\exp(\omega X+ \varphi V+ \psi \lambda+\chi).$$
Combining (\ref{affine}) with (\ref{PIDEt}), we know that $C(\tau;q),D(\tau;q),E(\tau;q)$ satisfy the system (\ref{odes}) with initial conditions (\ref{conditions}).
\end{proof}

\begin{remark}\label{r21}
Some exact expressions of $\Lambda(\tau;q)$ can be obtained in some special cases. In fact,
\begin{itemize}
\item[$(i)$] if the jump sizes $J^S$ and $J^V$ have independent asymmetric double exponential distributions with density functions
$$f(y)=p\eta_1e^{-\eta_1}\mathbf{1}_{\{y\geq 0\}}+q\eta_2e^{-\eta_2}\mathbf{1}_{\{y\leq 0\}},\quad \eta_1>1,\eta_2>0$$
and
$$f(y)=p^{,}\eta_3e^{-\eta_3}\mathbf{1}_{\{y\geq 0\}}+q^{,}\eta_4e^{-\eta_4}\mathbf{1}_{\{y\leq 0\}},\quad \eta_3>1,\eta_4>0,$$
respectively, where $p, q, p',q'\geq 0 $ with $p+q=1$ and $p'+q'=1$ represent the probabilities
of upward and downward jumps, respectively,  then the value of $\Lambda(\tau;q)$ in formula \eqref{e+} has the following form
$$\Lambda(\tau;q)=\left(\frac{p\eta_1}{\eta_1-1}+\frac{q\eta_2}{\eta_2+1}\right)^{\omega}
\left(\frac{p^,\eta_3}{\eta_3-1}+\frac{q^,\eta_4}{\eta_4+1}\right)^C-
\left(\frac{p\eta_1}{\eta_1-1}+
\frac{q\eta_2}{\eta_2+1}-1\right)\omega-1.$$

\item[$(ii)$] if $J^V\sim \exp(1/\eta)$ (exponential distribution with parameter rate $1/\eta$) and $J^S$ satisfies
$$J^S|J^V\sim N(\nu+\rho_J J^V,\delta^2),$$
that is,  the Gaussian distribution with mean $\nu+\rho_J J^V$ and variance $\delta^2$, then the value of $\Lambda(\tau;q)$ in formula \eqref{e+} has the following form
$$\Lambda(\tau;q)=\exp\left(\omega \nu+\frac{\delta^2\omega^2}{2}\right)\frac{1}{1-(\rho_J\omega+C)\eta}
-\left(\frac{e^{\nu+\rho_J J^V}}{1-\eta\rho_J}-1\right)\omega
-1,$$
where $\text{Re}((\rho_J\omega+C)\eta)<1$.
\end{itemize}
\end{remark}

\begin{remark}\label{r21}
We would like to mention that it suffices to use the marginal MGFs to price discretely sampled vanilla volatility swaps due to
their simpler payoff structure. Once the joint MGF is known, the respective marginal MGF can be obtained easily by setting the irrelevant parameters in the joint MGF to be zero. For example, the marginal MGF with respect to the state variable V can be
obtained by setting $\omega=\psi=\chi=0.$
\end{remark}

\begin{remark}
It is remarkable that when the stochastic intensity process is a constant ($\kappa_\lambda=\theta_\lambda=\sigma_\lambda$=0), the MGF can be found in \cite{Zheng} under this special case.  When the jump diffusion of variance process is removed, the MGF can be found in \cite{Huang} under this special condition. Moreover, if the stochastic intensity process is a constant and the jump diffusion of variance process is removed, the MGF can be gotten in \cite{Zhu15}.
\end{remark}
\section{Pricing Volatility Swaps}

In this section, we begin with an analytical solution approach to determine the fair price of the volatility swap which is modeled by \eqref{real}. Then we present the relationship between of discrete and continuous samples by our pricing approach.
\subsection{Volatility Swaps}

A volatility swap is a forward contract on realized historical volatility of the specified underlying equity index. The amount paid at expiration is based on a notional amount times the difference between the realized volatility and implied volatility. More specifically, assuming the current time is 0, the value of a volatility swap at expiry can be written as
$(RV-Kvol)\times L$, where the $RV$ is the annualized realized volatility over the contract life, Kvol is the annualized delivery price for the volatility swap, which is set to make the value of a volatility swap equal to zero for both
long and short positions at the time the contract is initially entered.
To a certain extent, it reflects market¡¯s expectation of the realized volatility in the future. The value $L$ is the notional amount of the swap in dollars per annualized volatility point squared and the realized volatility is always discretely sampled over a time period.

In this paper, we adopt the formulation of the realized volatility introduced by Barndorff-Nielsen and  Shephard (\cite{Barndorff-Nielsen0, Barndorff-Nielsen1,Barndorff-Nielsen2}). To start, we let $T$ as the total sampling period. Then we split equally the period to several fixed equal time interval $\triangle t$ and get $N$ different tenors $t_i=i\triangle t$ $(i=1,\cdots,N)$.
The realized volatility can be defined as follows:
$$RV=\sqrt{\frac{\pi}{2NT}}\sum_{i=1}^{N}\left|\frac{S_{t_i}-S_{t_{i-1}}}{S_{t_{i-1}}}\right|\times100.$$
As pointed out by Barndorff-Nielsen and  Shephard \cite{Barndorff-Nielsen0, Barndorff-Nielsen1, Barndorff-Nielsen2}, this definition is a more robust measurement of realized volatility have studied theoretical properties of realized volatility and obtained some closed-form solutions for pricing volatility derivatives under this formulation. By using risk neutral pricing theory, the value of RV can be obtained easily by
$$ K=E^{\mathbb{Q}} [RV]=\sqrt{\frac{\pi}{2NT}}\sum_{i=1}^{N} E^{\mathbb{Q}}
\left[\left|\frac{S_{t_i}-S_{t_{i-1}}}{S_{t_{i-1}}}\right|\right]\times100.$$

\begin{remark}\label{311} As mentioned by Windcliff et al. \cite{Windcliff}, there are at least two different measures of the realized volatility.
The realized volatility can be also defined as follows (\cite{Windcliff}):
$$RV^*=\sqrt{\frac{AF}{N}\sum_{i=1}^{N}\left(\frac{S_{t_i}-S_{t_{i-1}}}{S_{t_{i-1}}}\right)^2}\times100.$$
One can find that $RV\leq \sqrt{2/\pi}RV^*$ by utilizing the Cauchy inequality. This shows that the value $\sqrt{\pi/2}RV$ is a lower bound of $RV^*$.
\end{remark}

We would like to mention that $RV$ and $RV^*$ are slightly distinct when they are used to measure the realized volatility. In fact, the
definition of $RV^*$ is essentially calculated as the square root of average realized variance and it was used to term the volatility
swap contract by calculating realized volatility as a standard derivation swap \cite{Howison,Zhu15}. However,  $RV$ is the average of the realized volatility and it was employed to term this volatility swap as a volatility-average swap \cite{Howison,Zhu15}. For more work related to the realized volatility, we refer the reader to
\cite{Barndorff-Nielsen0,Barndorff-Nielsen1,Barndorff-Nielsen2,Broadie,Carr07,Carr13,Swishchuk} and the references therein.

Now, we turn to study the pricing problem for volatility swaps under the definition of $RV$. To obtain the value of volatility swaps, we first demonstrate the derivation of the characteristic function. Let us introduce a new stochastic variable as follows:
$$
Y_{t_{i-1},t_i}=\ln S_{t_i}-\ln S_{t_{i-1}}=X_{t_i}-X _{t_{i-1}}.
$$
Denote the probability density function of $Y_{t_{i-1},t_i}$ by $p(Y_{t_{i-1},t_i})$. Then $p(Y_{t_{i-1},t_i})$ can be easily obtained by the inverse Fourier transform with regard to the MGF. The probability of the event $Q_i:=P\{Y_{t_{i-1},t_i}>0\}$ can be carried out by following formula:
$$
Q_i=\int_0^\infty p(Y_{t_{i-1},t_i}) dY_{t_{i-1},t_i}=\frac{1}{2}+\frac{1}{\pi}\int_{0}^{\infty}\text{Re}\left[\frac{U(\triangle t,\omega i,V,\lambda)}{\omega i} \right]d\omega,
$$
where $\triangle t=t_i-t_{i-1}$. It is remarkable that the characteristic function $U(\triangle t,\omega i,V,\lambda)$ is obtained in each period $[t_{i-1},t_i]$. For each interval $[t_{i-1},t_i]$, we can use iteration methods to calculate the value of $U(\triangle t,\omega i,V,\lambda)$.
\begin{lemma}
Let
$$
q(Y_{t_{i-1},t_i})=\exp(Y_{t_{i-1},t_i}-r\triangle t)p(Y_{t_{i-1},t_i}).
$$
Then $q(Y_{t_{i-1},t_i})$ is a probability density function of the stochastic variable $Y_{t_{i-1},t_i}$.
\end{lemma}
\begin{proof}
Obviously, $q(Y_{t_{i-1},t_i})\geq0$. Since
$$E^{\mathbb{Q}}[\exp(-\triangle t)S_{t_i}/S_{t_{i-1}}]=1,$$
we have
$$\int_{-\infty}^{+\infty}q(Y_{t_{i-1},t_i})dY_{t_{i-1},t_i}=1.$$
This shows that $q(Y_{t_{i-1},t_i})$ is a probability density function of the stochastic variable $Y_{t_{i-1},t_i}$.
\end{proof}

\begin{lemma}
Let
$$
\widetilde{Q_i}=\int_0^\infty q(Y_{t_{i-1},t_i})dY_{t_{i-1},t_i}.
$$
Then
$$\widetilde{Q_i}=\frac{1}{2}+\frac{1}{\pi}\int_{0}^{\infty} \text{Re}\left[\frac{\exp(-r\triangle t)U(\triangle t,\omega i+1,V,\lambda)}{\omega i}           \right]d\omega.$$
\end{lemma}
\begin{proof}
Denote the corresponding characteristic function of $q(Y_{t_{i-1},t_i})$ by $\widetilde{U}(\triangle t,\omega ,V,\lambda)$.
By using Fourier transform with regard to the probability density function $q(Y_{t_{i-1},t_i})$, we have
\begin{eqnarray*}
\widetilde{U}(\triangle t,\omega ,V,\lambda)&=&\mathcal{F}\left[\exp(Y_{t_{i-1},t_i}-r\triangle t)p(Y_{t_{i-1},t_i})\right]\nonumber\\
&=&\exp(-r\triangle t)\mathcal{F}\left[\exp(Y_{t_{i-1},t_i})p(Y_{t_{i-1},t_i})\right]\nonumber\\
&=& \exp(-r\triangle t)\int_{-\infty}^{+\infty}\exp(i\omega Y_{t_{i-1},t_i})\exp(Y_{t_{i-1},t_i})p(Y_{t_{i-1},t_i})dY_{t_{i-1},t_i}\\
&=&\exp(-r\triangle t)U(\triangle t,\omega i+1 ,V,\lambda).
\end{eqnarray*}
It follows that
$$
\widetilde{Q_i}=\frac{1}{2}+\frac{1}{\pi}\int_{0}^{\infty}\text{Re}\left[\frac{\exp(-r\triangle t)U(\triangle t,\omega i+1 ,V,\lambda)}{\omega i}           \right]d\omega.
$$
This completes the proof.   \end{proof}

\begin{proposition}\label{prop31}
The fair strike value of volatility swap can be given as follows
$$K=\sqrt{\frac{\pi}{2NT}}\int_0^\infty\sum_{i=1}^{N}\text{Re}\left[\frac{U(\triangle t,\omega i+1 ,V,\lambda)-U(\triangle t,\omega i ,V,\lambda)}{\omega i} \right]d\omega\times100.$$
\end{proposition}
\begin{proof}
For each time interval $[t_{i-1},t_i]$, we have
\begin{eqnarray*}
E^{\mathbb{Q}} \left[\left|\frac{S_{t_i}}{S_{t_{i-1}}}-1\right|\right]
&=&\int_{0}^{\infty}\left| \exp(Y_{t_{i-1},t_i})-1 \right| p(Y_{t_{i-1},t_i}) dY_{t_{i-1},t_i}\nonumber\\
&=& \int_{0}^{+\infty}\left(\exp(Y_{t_{i-1},t_i})-1 \right) p(Y_{t_{i-1},t_i}) dY_{t_{i-1},t_i}\\
&& \mbox{}+\int_{-\infty}^{0}\left(1-\exp(Y_{t_{i-1},t_i})\right) p(Y_{t_{i-1},t_i}) dY_{t_{i-1},t_i}\nonumber\\
&=& -\int_{0}^{+\infty} p(Y_{t_{i-1},t_i}) dY_{t_{i-1},t_i}
+e^{r\triangle t}\int_{0}^{+\infty} q(Y_{t_{i-1},t_i}) dY_{t_{i-1},t_i}\\
&& \mbox{}+\int_{-\infty}^{0}p(Y_{t_{i-1},t_i})dY_{t_{i-1},t_i}+e^{r\triangle t}\int_{-\infty}^{0}q(Y_{t_{i-1},t_i})dY_{t_{i-1},t_i}\\
&=& 1-2\int_{0}^{+\infty} p(Y_{t_{i-1},t_i}) dY_{t_{i-1},t_i}\\
&& \mbox{}+e^{r\triangle t}\left(2\int_{0}^{+\infty}q(Y_{t_{i-1},t_i})dY_{t_{i-1},t_i}-1\right)\\
&=& \frac{2}{\pi}\int_0^\infty\text{Re}\left[\frac{U(\triangle t,\omega i+1 ,V,\lambda)-U(\triangle t,\omega i ,V,\lambda)}{\omega i} \right]d\omega.
\end{eqnarray*}

This leads to the final pricing formula for the volatility swaps in the following form:
\begin{eqnarray*}
K&=&E^{\mathbb{Q}}[RV]=\sqrt{\frac{\pi}{2NT}}\sum_{i=1}^{N} E^{\mathbb{Q}}
\left[\left|\frac{S_{t_i}-S_{t_{i-1}}}{S_{t_{i-1}}}\right|\right]\times100\\
&=&\sqrt{\frac{\pi}{2NT}}\sum_{i=1}^{N}\int_0^\infty\text{Re}\left[\frac{U(\triangle t,\omega i+1 ,V,\lambda)-U(\triangle t,\omega i ,V,\lambda)}{\omega i}\right]d\omega\times100\\
&=&\sqrt{\frac{\pi}{2NT}}\int_0^\infty\sum_{i=1}^{N}\text{Re}\left[\frac{U(\triangle t,\omega i+1 ,V,\lambda)-U(\triangle t,\omega i ,V,\lambda)}{\omega i}\right]d\omega\times100.
\end{eqnarray*}
This completes the proof. \end{proof}

\subsection{Convergence of pricing scheme}

To understand the properties of realized volatility and get the pricing formula under continuous samplings, we need to recall the well know connection between realized variance and quadratic variation. Especially, in this study, we progress to work a more general aggregate volatility measure. We recall the definition of realized power variation introduced by Barndorff-Nielsen. Considering mentioned above, we work on the time interval from $0$ to $T$, and assume we have observations every $\triangle t>0$ period of time. The representation of the $u$-th order variation process ($u>0$) is given by
\begin{equation}\label{rpv}
\{X\}^{[u]}(t)=p-\left(\lim_{\triangle t \downarrow 0}(\triangle t)^{1-u/2}\sum_{i=1}^{\left \lfloor T/\triangle t \right \rfloor} |Y_{t_{i-1},t_i}|^u\right).
\end{equation}
Here, for any real number $a$, $\left \lfloor a \right \rfloor $ denotes the largest integer than or equal to a. Note that the normalisation $(\triangle t)^{1-u/2}$ is essential in power variation. In detail, the normalisation is one and so disappears when $u=2$, the normalisation goes off to infinity when $u>2$ and goes to zero when $u<2$ as $\triangle t \downarrow 0$ . The key property of power variation for stochastic volatility model is given as allows.

\begin{lemma}\label{l3.3} Let
$$X^{(1)}_t=\int_0^t(r-d-\lambda_s-\frac{1}{2}V_s)ds+\int_0^t\sqrt{V_s}dW_s^S.$$
Then
$$
\{X^{(1)}\}_i^{[u]}(t)=\mu_r\int_0^t V^{\frac{r}{2}}_s ds,
$$
where
$$
\mu_r=\mathbb{E}|v|^u=2^{u/2}\frac{\Gamma(\frac{u+1}{2})}{\Gamma(\frac{1}{2})}
$$
for $u>0$ with $v\sim N(0,1).$
\end{lemma}
\begin{proof}
The proof is similar to the one of Barndorff-Nielsen and Shephard (\cite{Barndorff-Nielsen1}) and so we omit it here.
\end{proof}

Next we will see how power variation changes when jumps occur. Consider the log-price $X_t$ in our model such that $X_t=X^{(1)}_t+X^{(2)}_t$ with
\begin{equation}\label{X2}
X^{(2)}_t=\sum_{i=1}^{N_t}(e^{J^S_i}-1),
\end{equation}
where $N$ is a finite activity, simple counting process such that $N_t<\infty$ for all $t\geq 0$. Then the power variation of $X_t$ is reported by following lemma.
\begin{lemma}\label{L}
If $X^{(1)}_t$ and $X^{(2)}_t$ are independent and $0<u<2$, then the power variation of $X_t$ has the following form
$$ \{X\}^{[u]}(t)=\mu_r\int_{0}^{t}V^{\frac{u}{2}}_s ds.
$$
\end{lemma}
\begin{proof}
Let $X^{(2)}_{t_{i-1},t_i}=X^{(2)}_{t_i}-X^{(2)}_{t_{i-1}}$ and
\begin{equation}\label{JX}
J(X^{(2)}_s)=X^{(2)}_s-X^{(2)}_{s-},
\end{equation}
respectively. Taking $g(x)=x^u$ ($0<u<2$) in Theorem 1 of \cite{Woerner}, we have
$$\sum_{i=1}^{\left \lfloor T/\triangle t \right \rfloor} |X^{(2)}_{t_{i-1},t_i}|^u \overset{p}{\rightarrow}\sum\left(\left|J(X^{(2)}_s)\right|^u:0<s\leq t \right)$$
as $\triangle t \rightarrow 0$. Thus, it follows from \eqref{X2} and \eqref{JX} that
$$\sum\left(\left|J(X^{(2)}_s)\right|^u:0<s\leq t \right) \overset{p}{\rightarrow}  \sum_{i=1}^{N_t}\left|e^{J^S_i}-1 \right|^u$$
and so
$$\sum_{i=1}^{\left \lfloor T/\triangle t \right \rfloor} |X^{(2)}_{t_{i-1},t_i}|^u \overset{p}{\rightarrow}\sum_{i=1}^{N_t}\left|e^{J^S_i}-1 \right|^u.
$$
Since $\sum_{i=1}^{N_t}\left|e^{J^S_i}-1 \right|^u$ is a constant, as
$\triangle t \rightarrow 0$, one has
$$(\triangle t)^\beta \sum_{i=1}^{\left \lfloor T/\triangle t \right \rfloor} |X^{(2)}_{t_{i-1},t_i}|^u \overset{p}{\rightarrow} 0$$
for $\beta>0$. Hence, the power variation of $X^{(2)}$ is zero. Furthermore, since $X_t=X^{(1)}_t+X^{(2)}_t$, similar to the proof of Theorem 4 in \cite{Barndorff-Nielsen1}, we can show that
$$\sum_{i=1}^{\left \lfloor T/\triangle t \right \rfloor} |X_{t_{i-1},t_i}|^u =\sum_{i=1}^{\left \lfloor T/\triangle t \right \rfloor} |X^{(1)}_{t_{i-1},t_i}|^u+\mathcal{O}_p(N_t),$$
where $\mathcal{O}_p(\cdot)$ is the equivalent quantity with probability 1.
Thus, it follows from Lemma \ref{l3.3} that
$$(\triangle t)^\beta \sum_{i=1}^{\left \lfloor T/\triangle t \right \rfloor} |X_{t_{i-1},t_i}|^u \overset{p}{\rightarrow} \mu_r\int_{0}^{t}V^{\frac{u}{2}}_s ds$$
for $\beta=1-\frac{u}{2}$. This delivers the required result by the definition of $\{X\}^{[u]}(t)$.
\end{proof}
\begin{remark}
Note that the probability limit of realised power variation is unaffected by the presence of jumps when $0<u<2$. In particular, if the volatility process $V(t)$ is continuous (the jump term $J^S$ is removed),  then
$$
V(t)=\mu_r^{-1} \left(\frac{ \partial \{X\}^{[u]}(t)}{\partial t}\right)^{1/u}.
$$
\end{remark}

\begin{proposition}\label{prop32}
The fair continuous volatility strike can be given by the following formula
$$K=E^{\mathbb{Q}} [RV]=\frac{1}{2\sqrt{\pi}T}\int_0^T\int_0^\infty\frac{1-E[e^{-sV_t}]}{s^{3/2}}ds dt\times100.$$
\end{proposition}
\begin{proof}
We define the fair continuous volatility strike as such the present value of the contract at time zero is equal to zero. This corresponds to solving the equation
\begin{equation}\label{cont}
E^{\mathbb{Q}} [e^{-rT}(RV-K)]=0.
\end{equation}
Note that the realized volatility under continous samples can be expresses (let $r=1$ in Lemma \ref{L}) as
$$RV=
\lim_{N \rightarrow \infty}\sqrt{\frac{\pi}{2NT}}\sum_{i=1}^{N}\left|\frac{S_{t_i}-S_{t_{i-1}}}{S_{t_{i-1}}}\right|\times100
=\frac{1}{T}\int_{0}^{T}\sqrt{V_t}dt\times100.
$$
It follows from \cite{Schurger} that
$$\sqrt{X}=\frac{1}{2\sqrt{\pi}}\int_0^\infty \frac{1-e^{-sX}}{s^{3/2}}ds.$$
Taking expection on both sides and using Fubini's theorem we get
$$E^{\mathbb{Q}}[\sqrt{X}]=\frac{1}{2\sqrt{\pi}}\int_0^\infty \frac{1-E^{\mathbb{Q}}[e^{-sX}]}{s^{3/2}}ds,$$
where $E^{\mathbb{Q}}[e^{-sX}]$ is the characteristic function of the the stochastic variable $X$.
Choosing the $X$ above to the realized volatility we thus obtain a solution formula for the volatility strike price. Note that the formula (\ref{cont}) and using above Laplace transforms, we have that the strike price is given by
$$K=E^{\mathbb{Q}} [RV]=\frac{1}{2\sqrt{\pi}T}\int_0^T\int_0^\infty\frac{1-E^{\mathbb{Q}}[e^{-sV_t}]}{s^{3/2}}ds dt\times100.$$
This completes the proof.
\end{proof}

Here, the characteristic function of $V_t$ can refer to the Theorem \ref{theo} and be obtained by Remark \ref{r21}. Thus we need to use numerical integration techniques in order to solve above integral which yields the volatility strike price. This formula is very similar to the formula for the strike price in the case when no jumps were assumed \cite{Zhu15} although it is important to remember that the characteristics of the stochastic system are different.

\section{Numerical Examples}

We now investigate the impact of previous modelling assumptions and contractual designs on the fair values of volatility swaps. We begin by employing the effect of alternative assumptions about the stochastic process followed by the price of the underlying asset, the volatility process and the jump intensity process of Possion process. In particular, we give some numerical results under double exponential jump diffusion model.
\subsection{Baseline parameter value}

The main goal of our numerical analysis is to make clear how the jumps and intensity impact on the pricing volatility swaps. To make a reliable analysis, we take following baseline parameter values unless otherwise stated:
 $r=0.05$, $d=0.005$, $V_0=0.04$, $\rho=-0.64$, $\lambda_0=0.02$, $\sigma_V=0.6$, $\kappa_V=10$, $\theta_V=0.05$. These parameters are also adopted in \cite{Dragulescu} and \cite{Zhu15} by estimating the real market data. We show all numerical results in figures 1-5 and tables 1-2.

The stochastic processes in \eqref{real} can be discretized via the Monte Carlo methods studied in Section 8.2 in \cite{Schoutens} and so we can use Monte Carlo methods to price the volatility swaps. We also use the pricing formula in Propositions \ref{prop31} and \ref{prop32} to obtain the volatility swaps values. Some pricing results under different models are reported in \cite{Broadie,Carr07,Swishchuk,Zhu15} and we show these results in  Figure \ref{compare}. As shown in Figure \ref{compare}, volatility swap prices under various pricing formulas are distinct. We can see that, the results from Monte Carlo methods for \eqref{real} converge to our pricing formula, which confirm that our pricing formula is useful immediately. It is remarkable that the pricing formula from Zhu and Lian \cite{Zhu15} and our pricing formula are both under the definition of $RV$, and the volatility swaps values from our pricing formula are higher than \cite{Zhu15} because of the presence of jump risks, while \cite{Zhu15} employs Heston model. The results from Swishchuk \cite{Swishchuk} pricing formula gives the volatility swap value under continuous samplings, which are cheaper than prices from Zhu and Lian \cite{Zhu15} and our pricing formula. It may be caused the property of absolute value operator in the definition of $RV$.

In addition, pricing formulas form Carr \& Lee \cite{Carr07}, and Broadie \& Jain \cite{Broadie} are both studied from SVSJ model under the framework of definition $RV^*$. By comparisons between the results from various pricing formulas, we find that these values are higher than those under the definition of $RV$, which are consistent with the results from Remark \ref{311}. We also find the convergence from discrete to continuous samplings under two definition are distinct. The discrete-sampled volatility swaps price under $RV$ decreases progressively to the continuous-sampled price, while the convergence of price under $RV^*$ is opposite to $RV$.
\begin{figure}[H]
  \centering
  \includegraphics[height=0.5\textwidth]{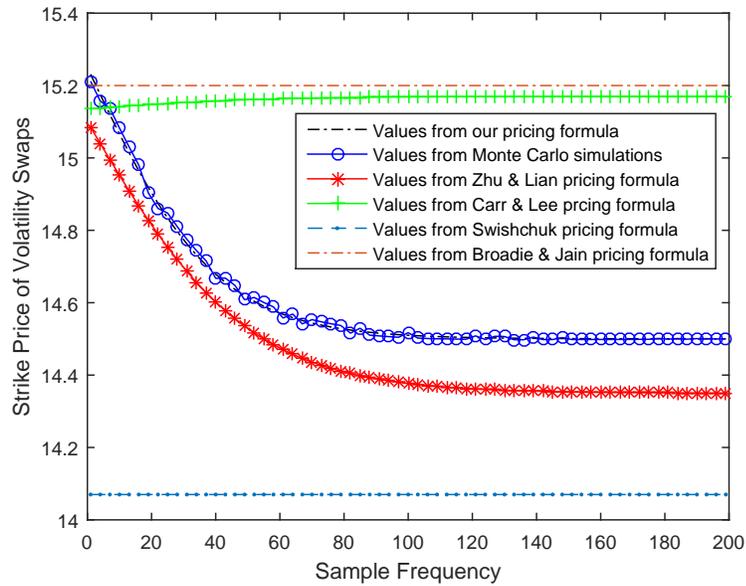}
  \caption{Strike price of volatility swaps with different pricing formulas in \cite{Broadie,Carr07,Swishchuk,Zhu15}}
  \label{compare}
\end{figure}
\subsection{Effect of Jumps}

\begin{figure}[H]
  \centering
  \includegraphics[height=0.5\textwidth]{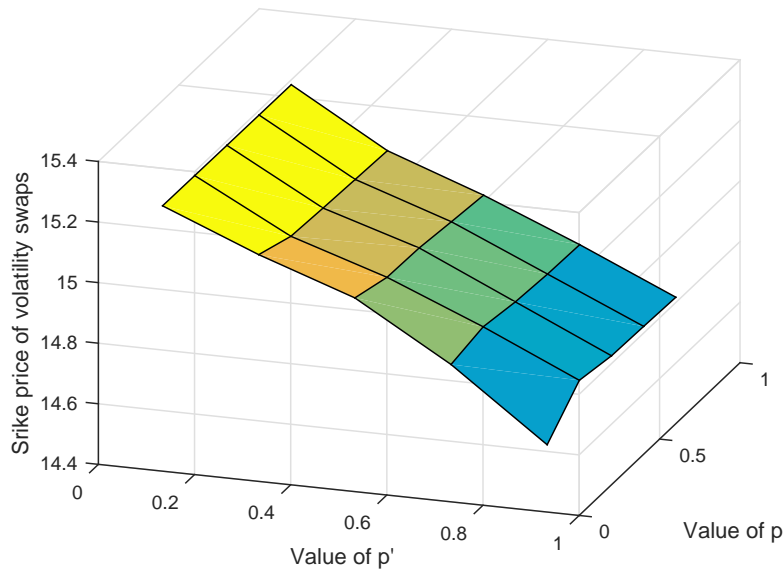}
  \caption{Strike price of volatility swaps with different $p$ and $p'$ values}
  \label{p}
\end{figure}

Firstly, we explore the jump effects to pricing volatility swaps under two independent double exponential jump diffusion model. More precisely, the jump sizes $J^S$ and $J^V$ have independent asymmetric double exponential distributions with density functions
$$f(y)=p\eta_1e^{-\eta_1}\mathbf{1}_{\{y\geq 0\}}+q\eta_2e^{-\eta_2}\mathbf{1}_{\{y\leq 0\}},\quad \eta_1>1,\eta_2>0$$
and
$$f(y)=p^{,}\eta_3e^{-\eta_3}\mathbf{1}_{\{y\geq 0\}}+q^{,}\eta_4e^{-\eta_4}\mathbf{1}_{\{y\leq 0\}},\quad \eta_3>1,\eta_4>0,$$
respectively, where $p, q, p',q'\geq 0 $ with $p+q=1$ and $p'+q'=1$ represent the probabilities
of upward and downward jumps, respectively.
\begin{table}[H]
\centering \caption{Strike price of volatility swaps with different $\eta_1$ and $\eta_2$ values}\label{tt1}
  \centering
\begin{tabular}{c c c c c c c c}
\hline
   &2.0 &3.0 &4.0 &5.0 &6.0 &7.0 &8.0\\ \hline
1.2 &15.0652 &15.0661 &15.0664 &15.0674 &15.0670 &15.0766 &15.1152\\
2.2 &15.0646 &15.0658 &15.0665 &15.0670 &15.0697&15.0763 &15.1152\\
3.2 &15.0638 &15.0652 &15.0663 &15.0665 &15.0694 &15.0760 &15.1149\\
4.2 &15.0634 &15.0641 &15.0659 &15.0663 &15.0691 &15.0756 &15.1145\\
5.2 &15.0640 &15.0634 &15.0647 &15.0665 &15.0687 &15.0750 &15.1150\\
6.2 &15.0632 &15.0639 &15.0635 &15.0656 &15.0675 &15.0742 &15.1154\\
7.2 &15.0614 &15.0612 &15.0630 &15.0637 &15.0664 &15.0727 &15.1160\\ \hline
\end{tabular}
\end{table}

\begin{table}[H]
\centering \caption{Strike price of volatility swaps with different $\eta_3$ and $\eta_4$ values}\label{tt2}
  \centering
\begin{tabular}{c c c c c c c c}
\hline
 &2.0 &3.0 &4.0 &5.0 &6.0 &7.0 &8.0\\ \hline
1.2 &15.1544 &15.1545 &15.1538 &15.1511 &15.1489 &15.1425 &15.1122\\
2.2 &15.1542 &15.1544 &15.1544 &15.1523 &15.1493 &15.1430 &15.1122\\
3.2 &15.1545 &15.1542 &15.1544 &15.1538 &15.1497 &15.1437 &15.1123\\
4.2 &15.1560 &15.1551 &15.1543 &15.1545 &15.1507 &15.1447 &15.1124\\
5.2 &15.1576 &15.1571 &15.1560 &15.1543 &15.1535 &15.1459 &15.1127\\
6.2 &15.1601 &15.1592 &15.1581 &15.1571 &15.1542 &15.1479 &15.1134\\
7.2 &15.1614 &15.1624 &15.1626 &15.1614 &15.1578 &15.1510 &15.1145\\ \hline
\end{tabular}
\end{table}
From Figure \ref{p}., we notice that with the increasing of $p$ or $p'$ values, the values of discrete
volatility swaps are decreasing. From Tables \ref{tt1} and \ref{tt2}, we can also observe that, when
the the probabilities
of upward and downward jumps changing, the value of volatility swaps are also fluctuating. This implies that
the parameters $\eta_1$ , $\eta_2$ , $\eta_3$  and $\eta_4$  have effect on the values of volatility swaps.
Finally, we can see that, when jump uncertainty is increasing, the value of volatility swaps is filling correspondingly. The implication
is that the jump diffusion can impact and change the value of a volatility swap, ignoring the effect of jumps will
result in miss-pricing. Working
out the analytical pricing formula for discretely-sampled volatility swaps can help pricing
volatiltiy swaps more accurately.

\subsection{Effect of Stochastic Intensity}

\begin{figure}[H]\centering
\includegraphics[width=.48\textwidth]{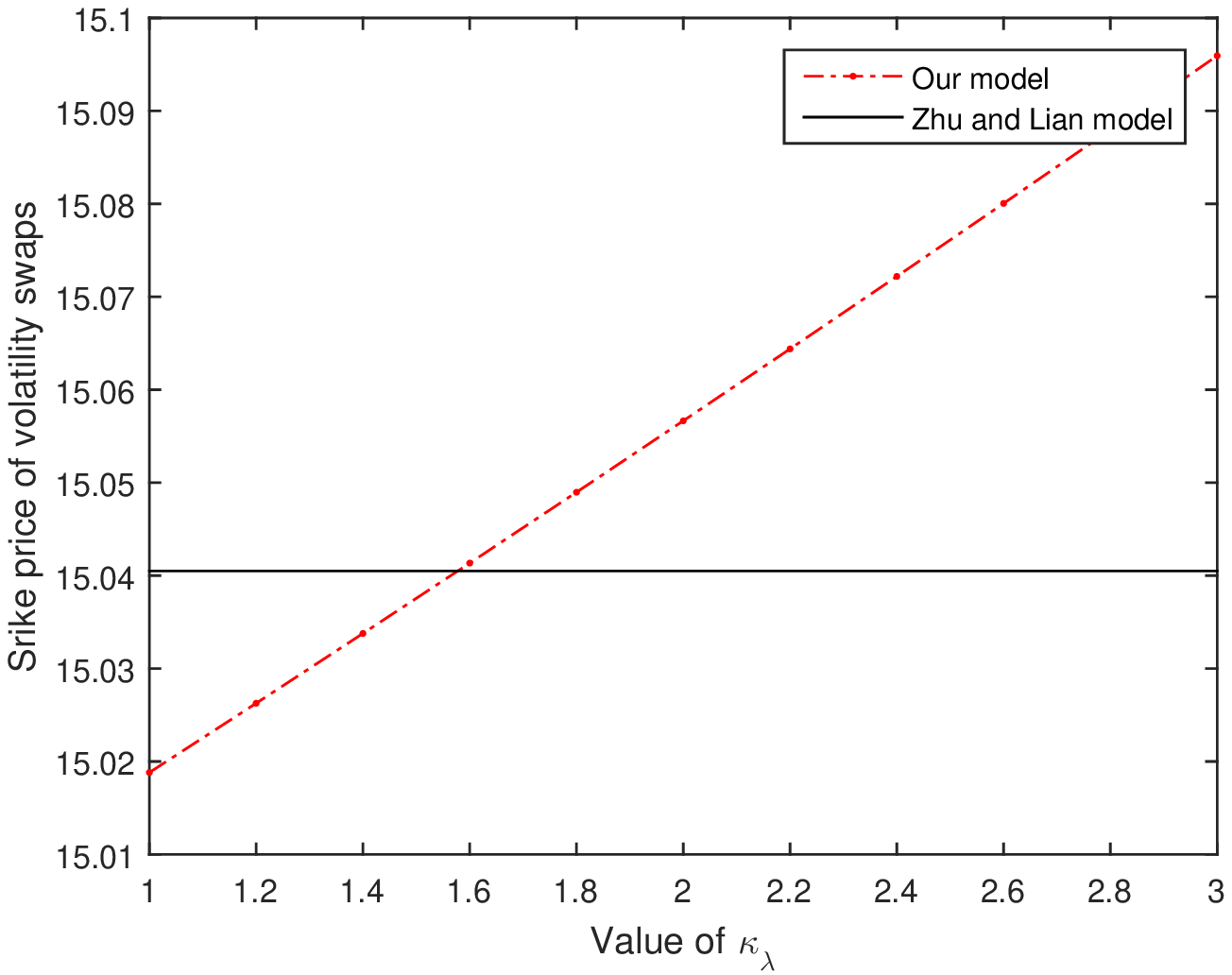}
\includegraphics[width=.48\textwidth]{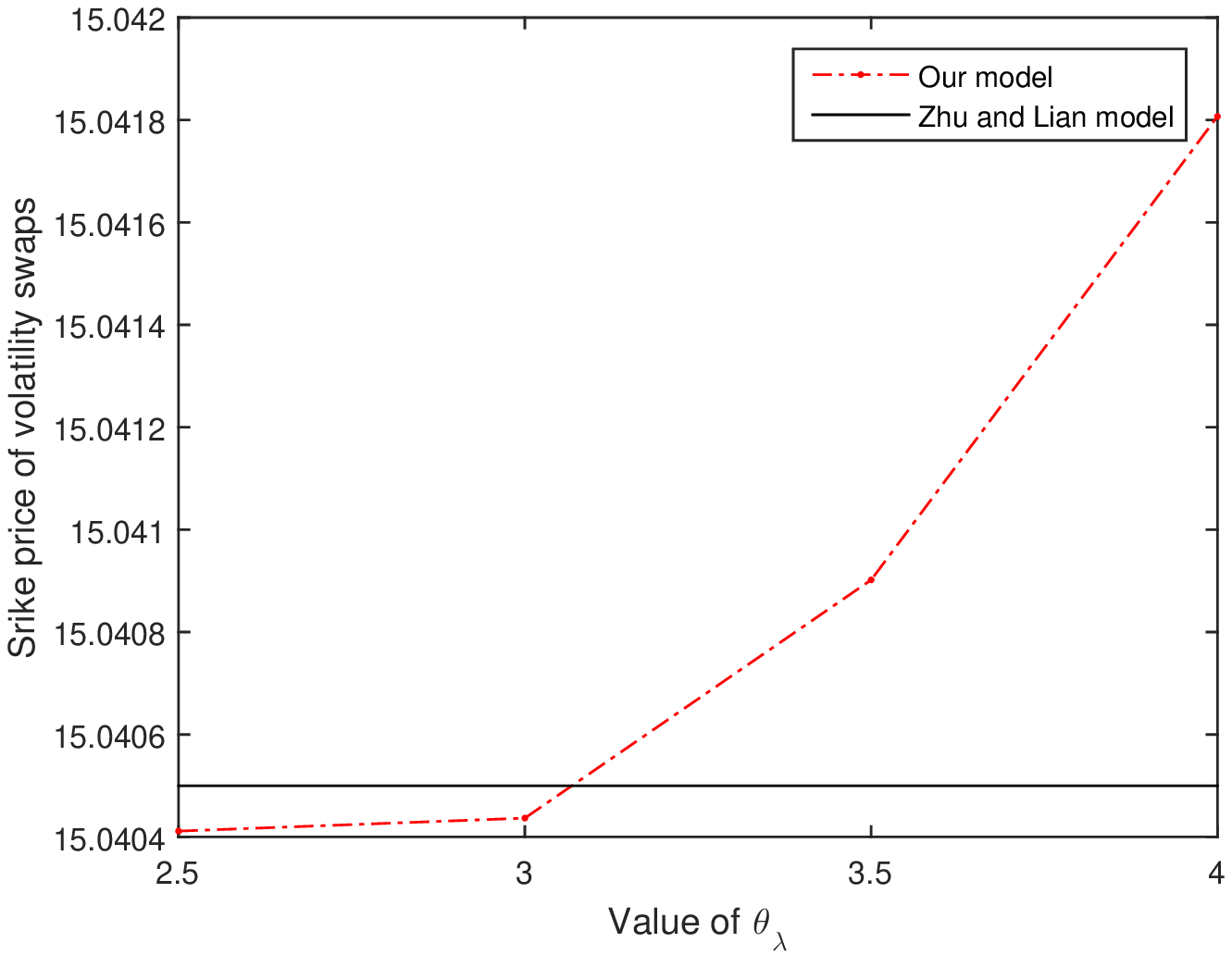}
\includegraphics[width=.48\textwidth]{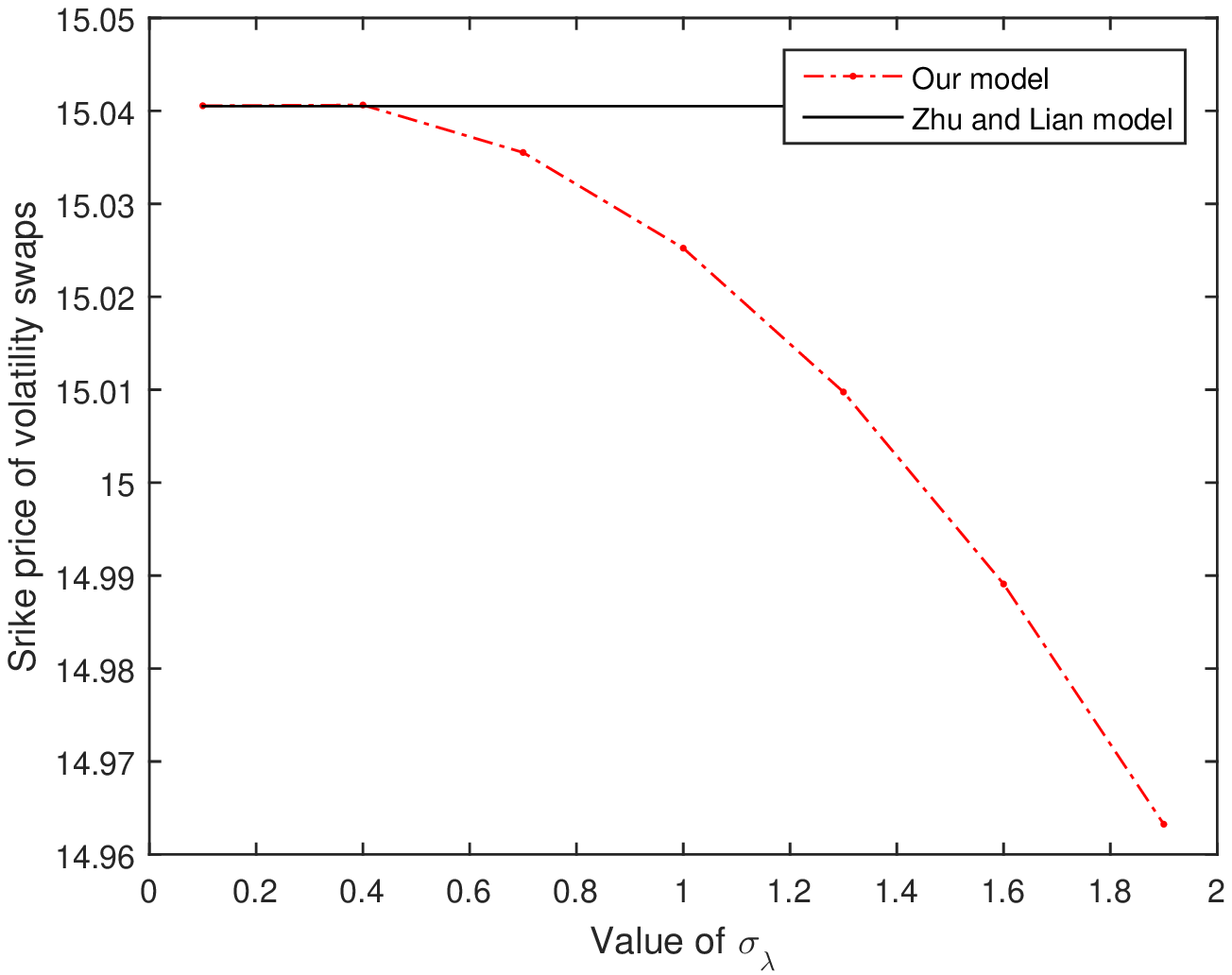}
\caption{Strike price of volatility swaps with different parameters}
  \label{para}
\end{figure}

We also test the sensitive effect of stochastic intensity to the strike price of volatility swaps. Figure \ref{para} shows the following facts: (i) the value of the volatility swap is increasing when $\kappa_\lambda$ is increasing; (ii) the value of the volatility swap is increasing when $\theta_\lambda$ is increasing; (iii)
the value of the volatility swap is decreasing when $\sigma_\lambda$ is increasing; (iv) the price change of the volatility swap is more sensitive with respective to
$\theta_\lambda$ and $\sigma_\lambda$.  Figure \ref{para} also shows that our numerical results are quite different from the ones of  Zhu and Lian \cite{Zhu15}.

Taken all together, we find that the volatility of intensity processes can impact and change the value of volatility swaps, ignoring the effect of the intensity will result in miss-pricing. In addition, it is notable that the sensitivities of the strike price of volatility swaps with respective to $\kappa_\lambda$, $\theta_\lambda$ and $\sigma_\lambda$ are quite different.  Therefore, to consider the effect of the intensity is very essential for pricing volatility swaps.

\section{Conclusions}
The main purpose of this paper is to propose a new stochastic volatility model with jumps and stochastic intensity and study the volatility swaps valuation problem described by this model. By using Feynman-Kac theorem, we deliver the joint moment generating function of this model via a partial integral differential equation.
Moreover, we derive the discrete and continuous sampled volatility swap pricing formulas by employing transform techniques and show the relationship between two pricing formulas. The contributions of this paper can be summarized as follows: (i) proposes stochastic volatility with jumps and stochastic intensity model at the first time; (ii) derives the joint moment generating function of this model by using the affine structure method introduced by Duffie \cite{Duffie} and the results presented in Yang et al. \cite{Yang18}; (iii) gives the pricing formula for discrete and continuous samples, respectively; (iv) shows the impacts of jumps and stochastic intensity on the fair strike price of volatility swaps.

Though this paper is focused on pricing the volatility swaps under the stochastic volatility model with jumps and stochastic intensity, the analytical procedures can be employed to some other models for the underlying asset price.  Therefore, the method developed in this paper can be extended to other pricing problems in connection to L\'evy processes with other stochastic volatilities, such as a VG process with the GARCH volatility. We leave these problems for our future work.


\begin{thebibliography}{99}
\bibitem{Barndorff-Nielsen0}O.E. Barndorff-Nielsen, N. Shephard, Realized power variation and stochastic volatility models, {\it Bernoulli} {\bf 9 (2)} (2003), 243-265.
\bibitem{Barndorff-Nielsen1}O.E. Barndorff-Nielsen, N. Shephard, Power and bipower variation with stochastic volatility and jumps, {\it J. Financ. Econ.} {\bf 2 (1)} (2004), 1-37.
\bibitem{Barndorff-Nielsen2}O.E. Barndorff-Nielsen, N. Shephard, Econometrics of testing for jumps in financial economics using bipower variation, {\it J. Financ. Econ.} {\bf 4 (1)} (2006), 1-30.
\bibitem{Bates}D.S. Bates, U.S. stock market crash risk, 1926-2010, {\it J. Financ. Econ.} {\bf 105 (2)} (2012), 229-259.
\bibitem{Broadie}M. Broadie, A. Jain, The effect of jumps and discrete sampling on volatility and variance swaps, {\it Int. J. Theor. Appl. Finance} {\bf 11 (8)} (2008) 761-797.
\bibitem{Carr98}P. Carr, D. Madan, R. Jarrow, Towards a theory of volatility trading, Volatility: New Estimation Techniques for Pricing Derivatives, Risk Publications (1998), 417-427.
\bibitem{Carr} P. Carr, R. Lee, Volatility derivatives, {\it Ann. Rev. Financ. Econ.} {\bf 1} (2009), 319-339.
\bibitem{Carr2}P. Carr, L. Wu, What type of process underlies options? A simple robust test, {\it J. Finance} {\bf 58 (6)} (2003), 2581-2610.
\bibitem{Carr07}P. Carr, R. Lee, Realised volatility and variance: options via swaps, {\it Risk} (2007), 76-83.
\bibitem{Carr13}P. Carr, L. Roger, Variation and share-weighted variation
swaps on time-changed L\'evy processes. {\it Finance Stoch.} {\bf 17}
(2013), 685-716.
\bibitem{Chang}C. Chang, C.D. Fuh, S.K. Lin, A tale of two regimes: theory and empirical evidence for a Markov-modulated jump diffusion model of equity returns and derivative pricing implications, {\it J. Banking Finance} {\bf 37(8)} (2013), 3204-3217.
\bibitem{Chr}G. Christian, S. Razvan, Derivative pricing with Wishart multivariate stochastic volatility, {\it J. Bus. Econ. Stat.} {\bf 28(3)} (2010), 438-451.
\bibitem{Cont13}R. Cont, T. Kokholm, A consisent pricing model for index options and volatility derivatives, {\it Math. Finance} {\bf 23(2)} (2013), 248-274.
\bibitem{Cont04}R. Cont, P. Tankov, {\it Financial Modelling with Jump Processes.} Boca Raton, FL: Chapman \& Hall/CRC, 2004.
\bibitem{Dragulescu}A.A. Dr\^agulescu, V.M. Yakovenko, Probability distribution of returns in the Heston model with stochastic volatility. {\it Quant. Finance} {\bf 2(6)} (2002), 443-453.
\bibitem{Duffie}D. Duffie, J. Pan, K. Singleton, Transform analysis and asset pricing for affine jump-diffusions, {\it Econometrica} {\bf 68 (6)} (2000), 1343-1376.
\bibitem{Gerhard}F. Gerhard, A survey of nonsymmetric Riccati euqtions, {\it Linear Algebra Appl.} {\bf 351-352} (2002), 243-270.
\bibitem{Howison}S. Howison, A. Rafailidis, H. Rasmussen, On the pricing and hedging of volatility derivatives, {\it Appl. Math. Financ.} {\bf 11} (2004) 317-346.
\bibitem{Huang}J.X. Huang, W.L. Zhu, X.F. Ruan, Option pricing using the fast Fourier transform under the double exponential jump model with stochastic volatility and stochastic intensity, {\it J. Comput. Appl. Math.} {\bf 263} (2014), 152-159.
\bibitem{Pun}C.S. Pun, S.F. Chung, H.Y. Wong, Variance swap with mean reversion, multifactor stochastic volatility and jumps, {\it Eur. J. Oper. Res.} {\bf 245 (2)} (2015), 571-580.
\bibitem{Schoutens}W. Schoutens, {\it L\'evy Processes in Finance: Pricing Financial Derivatives}, John Wiley \& Sons Inc, 2003. 
\bibitem{Schurger}K. Sch\"{u}rger, Laplace Transforms and Suprema of Stochastic Processes. Advances in Finance and Stochastics, Springer Berlin Heidelberg, (2002), 285-294.
\bibitem{Santa}P. Santa-Clara, S. Yan, Crashes, volatility, and the equity premium: lessons from S\&P 500 options, {\it Rev. Econ. Stat.} {\bf 92(2)} (2010), 435-451.
\bibitem{Swishchuk}A. Swishchuk, Modeling of variance and volatility swaps for financial markets with stochastic volatilities, Wilmott Mag. (September Issue) (2004)
64-72. Technical article.
\bibitem{Yang18}B.Z. Yang, J. Yue, N.J. Huang, Variance swaps under L\'{e}vy process with stochastic volatility and
  stochastic interest rate in incomplete market, arXiv:1712.10105[q-fin.PR].
\bibitem{Windcliff} H. Windcliff, P. Forsyth, K. Vetzal, Pricing methods and hedging strategies for volatility derivatives, {\it J. Bank. Finance} {\bf 30} (2006), 409-431.
\bibitem{Woerner}J. Woerner, Variational sums and power variation: a unifying approach to model selection and estimation in semimartingale models, {\it Statistics $\&$ Decisions} {\bf 21} (2003), 47-68.
\bibitem{Zheng}W.D. Zheng, Y.K. Kwok, Closed form pricing formulas for discretely sampled generalized variance swaps, {\it Math. Finance} {\bf 24(4)} (2014), 855-881.
\bibitem{Zhu11}S.P. Zhu, G.H. Lian, A closed-form exact solution for pricing variance swaps with stochastic volatility, {\it Math. Finance} {\bf 21 (2)} (2011), 233-256.
\bibitem{Zhu15}S.P. Zhu, G.H. Lian, Analytical pricing volatility swaps under stochastic volatility, {\it J. Comput. Appl. Math.} {\bf 288} (2015), 332-340.

\end{thebibliography}
\end{document}